\documentclass[12pt]{article}
\usepackage{amsmath,amssymb,amsfonts,amsthm,bm,slashed,graphics,graphicx}
\usepackage{epsfig,hyperref,subfigure,dsfont,multirow,psfrag,extarrows}
\usepackage{tikz}
\usepackage{mathtools}
\usepackage{pgfplots}

\DeclareMathOperator{\Tr}{Tr}

\newtheorem{lemma}{Lemma}
\newtheorem{remark}{Remark}
\newtheorem{definition}{Definition}
\newtheorem{theorem}{Theorem}

\newcommand{\bea}{\begin{eqnarray}}
\newcommand{\ea}{\end{eqnarray}}
\newcommand{\eea}{\end{eqnarray}}
\newcommand{\beq}{\begin{equation}}
\newcommand{\eeq}{\end{equation}}
\newcommand{\be}{\begin{equation}}
\newcommand{\ee}{\end{equation}}

\def\R{{\mathbb R}}

\newcommand{\cG}{{\cal{G}}}

\newcommand{\cF}{{\cal{F}}}
\newcommand{\cB}{{\cal B}}

\newcommand{\cJ}{{\cal J}}

\newcommand{\cD}{{\cal D}}

 \newcommand{\cW}{{\cal W}}
 
 \newcommand{\mJ}{{\mathbb J}}

\newcommand{\fm}{{\mathcal K}}

\newcommand{\cV}{{\cal{V}}}

\newcommand{\cS}{{\cal{S}}}
\newcommand{\bee}{\begin{equation}}

\newcommand{\ben}{\begin{eqnarray}}
\newcommand{\en}{\end{eqnarray}}

\newcommand{\prf}{{\noindent \bf Proof\; \; }}

\pgfplotsset{compat=1.16}
\begin{document}
\title{Multiscale Loop Vertex Expansion\\ for Cumulants, 
the $\phi^4_2$ Model}
\author{V. Rivasseau\\ Universit\'e Paris-Saclay, CNRS/IN2P3\\ IJCLab, 91405 Orsay, France}

\date{} 

\maketitle

\begin{abstract}  

We consider the $\phi^4_2$ model. It has been called the simplest non-trivial super-renormalisable model.
The method we use is the multiscale loop vertex expansion, an improvement of constructive field theory.  
We prove analyticity and Borel summability of the cumulants up to a \emph{finite order}.
\end{abstract}

\noindent\textbf{keywords}
Cumulants; Constructive Field Theory


\medskip\noindent
All data are available within the article.

\section{Introduction}

For a general exposition of constructive field theory, see \cite{Sim,GJ,Riv}.
The loop vertex expansion (LVE) was introduced in 2007  as a new tool in constructive field theory \cite{LVE,MagRiv}. 
The essential ingredients of LVE are the Hubbard-Stratonovich intermediate field representation \cite{Hub,Str}, the replica method \cite{MPV} and the BKAR formula \cite{BK,AR1}.

\medskip
A main feature of the LVE is that it is written in terms of trees which are exponentially bounded.
It means that the outcome of the LVE is convergent whereas the usual perturbative expansion \emph{diverges}.
For a review of the LVE, we suggest \cite{GRS}; for the actual mechanism of replacing 
Feynman graphs, which are not exponentially bounded, by trees, see \cite{RiZh1}.
For the LVE applied to cumulants when renormalisation is {\it absent}, see  \cite{GuKra,Riv1}.
\medskip
In order to deal with quantum fields when renormalisation is {\it present}, 
we aimed at combining the LVE with the multiscale approach 
that is synthesized in the book \cite{Riv}. 
 
\medskip
We have already performed the initial steps on this road. In \cite{MLVE} we presented a simple combinatorial model with renormalization. 
It is a model of conjugate vector fields with a quartic interaction and a particular propagator.
The main result of \cite{MLVE} is that this divergence is renormalized by using a Wick-ordered interaction. 
Then this multiscale LVE (MLVE) has been successfully applied to various general super-renormalisable fields of increasing complexity \cite{RiZh2,DR,DR1,RV}.
But the initial articles \cite{MLVE,RiZh2,DR,DR1,RV} are restricted to the partition function and its logarithm (the free energy) and do not include cumulants (connected Schwinger functions in the terminology of field theorists). The article \cite{Zhao} is more precise in the 
combinatorics of the various constants, but it is limited to the model $\phi^4_{d=1}$.

\medskip
We choose the $\phi^4_2$ model in the formalism of \cite{RiZh2} as an initial benchmark for studying cumulants in quantum field theory because it is the initial constructive model when renormalisation is present,
and the main problem of this model compared to the model of \cite{MLVE} is that the action is not positive. Remark also that, in this article, our cumulants depend solely on $J$, not on $J$ and $\bar J$, the source fields of \cite{GuKra}. 

\medskip
This paper is devoted to recovering classical results \cite{nelson,GJ1, EMS} in the LVE representation with finite cumulants.
Also, our view is different. It is not so much the constructive aspect and the LVE that we are after.
Our cumulants cannot be defined constructively but only perturbatively.
Therefore, we are satisfied with a theorem about Borel summability for cumulants \emph{up to a finite order}.

\section{The Model}
\label{sec-model}

\subsection{Formalism of \texorpdfstring{\cite{RiZh2}}{RiZh2}}

\begin{itemize}
\item We adopt the convention $\sum_{\varnothing}=0, \: \prod_{\varnothing}=1$.

\item Our convention for a norm is the $\mathbf{Euclidean}$ norm, i.e. if the real vector function $\mathbf f = (f_1 \cdots f_n)$ 
depends on {\it finite} variables $(a_1, \cdots, a_n)$, the norm in question  
can be expressed as the square root of the inner product of the associated vector and itself: 
\bea \label{eq0.5.0}
\Vert  \mathbf f \Vert = \sqrt{f_1^2 + \cdots +f_n^2}\;.
\ea

\end{itemize}

We choose our model in the formalism of \cite{RiZh2}, i.e.
we consider the Bosonic $\phi^4_2$ theory in a fixed volume, namely the unit square $[0,1]^2$, and with coupling constant $\lambda$. From now on, to reflect the fact that we are in dimension $d=2$, we note $x= (x_1,x_2)$. Any spatial 
integral has to be understood as restricted to $[0,1]^2$ and we use the notation $\Tr$ to mean 
\bee\int_{[0,1]^2} d^2x =\int_{[0,1]} dx_1\int_{[0,1]} dx_2, \quad x= (x_1,x_2).
\ee
The formal partition function of the theory with source $J(x)$ is
\begin{equation}\label{model}
Z(J,\lambda)=\int d\mu_{C_0} ( \phi ) e^{\Tr J(x) \phi (x)  -\frac{\lambda}{2}\Tr\; \phi^4(x)}
\end{equation}
where $d\mu_C$ is the normalized Gaussian measure with covariance or propagator 
\begin{equation}
C_0(x, y)=\frac{1}{4\pi}\int_0^\infty\frac{d\alpha}{\alpha}e^{-\alpha m^2-\frac{(x-y)^2}{4\alpha}},
\end{equation}
with $x$ and $y$ restricted to $[0,1]^2$, hence with free boundary conditions\footnote{The theory could equally well be considered
on the two-dimensional torus with periodic boundary conditions without significant change in our analysis.}.

A main problem in quantum field theory is to compute the generating function 
of the connected Schwinger functions.

The covariance at coinciding points $ C_0(x, x) = T$ corresponds to a self-loop or
tadpole in perturbation theory. It diverges logarithmically in the ultraviolet cutoff and
is the only primitive ultraviolet divergence of the theory. Renormalization reduces to Wick ordering; hence
the  renormalized model has partition function:
\bea\label{e5}
Z(J,\lambda)&=&\int d\mu_{C_0} ( \phi ) e^{\Tr J(x) \phi (x) -\frac{\lambda}{2}\Tr\; :\phi^4(x):} \\
&=& \int d\mu_{C_0} e^{\Tr J(x) \phi (x) -\frac{\lambda}{2}\Tr\;   [(\phi^2-3T)^2  -6 T^2]}, \label{e6}
\eea
where the Wick ordering in $:\phi^4(x): \equiv \phi^4-6T \phi^2 + 3T^2$ is taken with respect to $C_0$.

These expressions are formal, and to define the theory one needs to introduce an ultraviolet cutoff.
This is most conveniently done in a multi-scale representation \cite{Riv} which
slices the propagator in the parametric representation, then keeps only a finite number of slices and defines the ultraviolet cutoff as a maximal slice index $j_{max}$. The ultraviolet limit corresponds to $j_{max} \to \infty$. 
We are not interested in the infrared behavior of this  $\phi^4_2$ theory,
that is why we define\footnote{Beware we choose the convention of \emph{lower} indices for slices, as in \cite{MLVE}, not upper indices as in \cite{Riv}.}
\begin{definition}
 Let us define $M\in \R_+$, $M > 1$ and  
\bea
C  &:=& C_0-C_0^0 :=  C_0-\frac{1}{4\pi}\int_1^\infty\frac{d\alpha}{\alpha}e^{-\alpha m^2-\frac{(x-y)^2}{4\alpha}},   
\\ C_j &:=& \frac{1}{4\pi}\int_{M^{-2i}}^{M^{-2(i-1)}}\frac{d\alpha}{\alpha}e^{-\alpha m^2-\frac{(x-y)^2}{4\alpha}},
\quad C_{\le j_{max}}  :=\sum_{j=1}^{j_{max}} C_j \, .
\ea
\end{definition}

\subsection{Slices}

We put $m=0$ in the rest of this article since we are only interested in the ultraviolet behavior of this  $\phi^4_2$ theory.

\begin{lemma}
\label{Lem1}
$M$ can be chosen  so that 
\bea
C_j&\le& K e^{-c M^j|x-y|} \;= \;  e^{-\frac{e^{-2\pi}}{2} M^j|x-y|} .\label{multbound}
\ea
\end{lemma}
\prf
The first part of this lemma is proved in Lemma II.1.1. of \cite{Riv}. For the second part,
by the fundamental upper bound of an integrable function
\bea
\frac{1}{4\pi}\int_{M^{-2i}}^{M^{-2(i-1)}}\frac{d\alpha}{\alpha}e^{-\frac{(x-y)^2}{4\alpha}}
&\le& \frac{\log M}{2\pi} \, e^{- \frac14 M^{2(i-1)} (x-y)^2};
\ea
Since $x$ and $y$ are restricted to $[0,1]^2$,  hence $(x-y)^2\le 2$, $|x-y| \le \sqrt2$.
If $z^2$ is positive, $e^{-z^2} \le e^{-z}$, hence
\bea
 e^{- \frac14 M^{2(i-1)} (x-y)^2} &\le& e^{-\frac{M^j |x-y| }{2M}}.
\ea
We arrive at
\bea
\frac{1}{4\pi}\int_{M^{-2i}}^{M^{-2(i-1)}}\frac{d\alpha}{\alpha}e^{-\frac{(x-y)^2}{4\alpha}}
&\le&  \frac{\log M}{2\pi} \, e^{-\frac{M^j |x-y| }{2M}} = \, e^{-
\frac{e^{-2\pi}}{2} 
M^j|x-y|} 
\ea
if $M= e^{2\pi}$. \qed

\medskip
We put $M=e^{2\pi}$ in the rest of this article. 
Therefore, taking into account Lemma 1, 
\bea\label{e10}
T_j :=C_j (x,x), \quad T_{\le j} &:=& \sum_{k=1}^{j}  T_k , \quad
T_{\le j} = j , \quad T_{\le j_{max}} = j_{max} .
\ea

The partition function $Z(\lambda,j_{\max} ,J)$ obtained by substituting $C_{\le j_{max}}$ instead of $C_0$ in \eqref{model} and $T_{\le j_{max}}$
instead of $T$ in \eqref{e6} is now well defined:
\begin{definition}\label{Generating function of the moments}
The renormalized partition function {\it with sources} $J(x)$ is
\bea
Z(\lambda,j_{\max} ,J)&=&   \int d\nu(\sigma) e^{- \Tr\big\{ \log_2 
[1+2i \sqrt{\lambda} C_{\le j_{max}}\sigma ]\big\}}\nonumber \\&&  e^{- \Tr \big\{3\lambda  T_{\le j_{max}}^2  + 2i \sqrt{\lambda} T_{\le j_{max}}\big\}}
e^{- <J(y) R(\sigma) J(x)> }\label{eqp1}
\ea
with $R(\sigma)$ being defined by 
\bee \label{eq272} R(\sigma) \equiv [1+2i \sqrt{\lambda} C_{\le j_{max}}\sigma ]^{-1}.
\ee
This definition of $Z(\lambda,j_{\max} ,J)$ holds only as a \emph{power series in $\lambda$}.
\end{definition}
Now let us define \emph{cumulants}, similar to that one in \cite{GuKra}, Definition 2, but taking into account the symmetry.
The cumulants of $\phi_2^4$ of odd order vanish \cite{RiZh2}; hence we need to worry only about the cumulants of even order.
The cumulants of even order are, by a formula combining \cite{GuKra} and \cite{RiZh2}\footnote{Our notation, in all that is concerned with the cumulants, is in line with \cite{Riv1}. }:
\begin{definition}\label{cum0}
\bea\label{cum1}
\mathfrak{K}^{2\fm}_{j_{\max}}(\lambda, x_1, ..., x_{2\fm})\!
&:=&\!\Big[{\frac{\partial^{2\fm}}{\partial J(x_1)\cdots\partial J(x_{2\fm})}}\log{ Z(\lambda,j_{\max},J)} \Big]_{\{J\}=0}.
\ea
\end{definition}

\subsection{Single Slice Intermediate Field Representation}
\label{SinSlice}
In \cite{MagRiv} $d\nu$ is defined to be the ultra-local measure on $\sigma$ with covariance $\delta(x-y)$.
Then the following Theorems are stated and proved in \cite{MagRiv} : 
\begin{theorem}
\begin{eqnarray}\label{treeformul}
\lim_{\Lambda \to {\mathbb R}^4}\frac {\log Z(\Lambda)} {\vert \Lambda\vert } 
&=& \sum_{n=1}^{\infty} \frac{1}{n!}\sum_\cJ        \bigg\{ \prod_{\ell\in \cJ}  
\big[ \int_0^1 dw_\ell   \big]\bigg\} G_T(\sigma, x_{\ell_0})\vert_{x_{\ell_0} =0}  \\
G_\cJ(\sigma, x_{\ell_0})&=&\prod_{\ell\in \cJ}   \int d^4 x_\ell d^4 y_\ell 
\int  d\nu_\cJ (\{\sigma_v\}, \{ w \})  \nonumber \\
\hskip-1cm && \bigg\{ \prod_{\ell\in \cJ} \big[ \delta (x_\ell - y_\ell) 
 \frac{\delta}{\delta \sigma_{v(\ell)}(x_\ell)}\frac{\delta}{\delta \sigma_{v'(\ell)}(y_\ell)} 
\big] \bigg\} \prod_v V_v , \label{gt}
\end{eqnarray}
where 
\begin{itemize}

\item each line $\ell$ of the tree joins two different vertices $V_{v(\ell)}$ and $V_{v'(\ell)}$ 
at point $x_{\ell}$ and $y_{\ell}$, which are identified through the function
$\delta (x_\ell - y_\ell) $ (since the covariance of $\sigma$ is ultra-local),

\item the sum over $\cJ$ is over rooted trees over $n$ vertices, which have therefore $n-1$ lines, with root $\ell_0$,\footnote{The notation $T$ is replaced by $\cJ$ to be in line of \cite{RiZh2}.} 

\item the normalized Gaussian measure $d\nu_\cJ (\{\sigma_v\}, \{ w \})  $ over the vector field $\sigma_v$ has covariance
$$<\sigma_v,\sigma_{v'}>=
\delta (x-y) w^\cJ (v, v', \{ w\})$$ where $w^\cJ (v, v', \{ w\})$ is 1 if $v=v'$,
and the infimum of the $w_\ell$ for $\ell$ running over the unique path from $v$ to $v'$ in $\cJ$
if $v\ne v'$. This measure is well-defined because the matrix $w^\cJ$ is positive.

\end{itemize}
\end{theorem}

 To obtain the connected functions with
external legs,  we need to add resolvents to the initial loop vertices,
where a resolvent is an operator 
\bea
C_j(\sigma_r, x, y ), \quad C_j(\sigma) = D_j  \frac{1}{1+i H } D_j ,
\ea see \cite{MagRiv} for further details.   

\begin{theorem}\label{cumulants}
The cumulants of even order restricted to a single slice are given by a formula of \cite{MagRiv}, Equation (14), but with this article's notation:
\begin{eqnarray}\label{treeformulext} 
\mathfrak{K}^{2\fm}_{j}(\lambda, x_1, ..., x_{2\fm})
&=& \sum_{n=1}^{\infty}\frac{1}{n!} \sum_\cJ \bigg\{ \prod_{\ell\in \cJ}  
\big[ \int_0^1 dw_\ell \int d^2 x_\ell d^2 y_\ell \big]\bigg\}  \nonumber \\
&&\hskip-3.5cm  \prod_{r=1}^{\fm} \bigg\{ \ \int  d\nu_\cJ (\{\sigma_v\}, \{\sigma_r\}, \{ w \}) 
\prod_{\ell\in \cJ} \big[ \delta (x_\ell - y_\ell) 
 \frac{\delta}{\delta \sigma_{v(\ell)}(x_\ell)}\frac{\delta}{\delta \sigma_{v'(\ell)}(y_\ell)} 
 \big]\nonumber \\
 && \hskip-3.5cm  \prod_v V_v  \sum_{\pi 
 \text{ pairings $(x_{\pi(r,1)}, x_{\pi(r,2)})$}}  C_j
 (\sigma_{r}, x_{\pi(r,1)}, x_{\pi(r,2)}) \bigg\} \; ,
\end{eqnarray}
where 
\begin{itemize}

\item each line $\ell$ of the tree $\cJ$ joins two different loop vertices or resolvents
$V_{v(\ell)}$ and $V_{v'(\ell)}$ 
at point $x_{\ell}$ and $y_{\ell}$, which are identified through the function
$\delta (x_\ell - y_\ell) $ because the covariance of $\sigma$ is ultra-local,

\item the sum is over trees $\cJ$ joining the $n+\fm$ loop vertices and resolvents, which have therefore $n+\fm-1$ lines,

\item the measure $d\nu_\cJ (\{\sigma_v\}, \{\sigma_r\}, \{ w \})  $ over the 
$\{\sigma\}$ fields has covariance \ \ $<\sigma_\alpha,\sigma_{\alpha'}>=
\delta (x-y) w^\cJ (\alpha, \alpha', \{ w\})$ where $w^\cJ (\alpha, \alpha', \{ w\})$ is 1 if $\alpha=\alpha'$ 
(where $\alpha, \alpha'\in \{v\}, \{r\}$),
and the infimum of the $w_\ell$ for $\ell$ running over the unique path from $\alpha$ to $\alpha'$ in $\cJ$
if $\alpha\ne \alpha'$. This measure is well-defined because the matrix $w^\cJ$ is positive,

\item the sum over $\pi $ runs over pairings into pairs $(x_{\pi(r,1)}, x_{\pi(r,2)}),r=1,...,\fm$.
\end{itemize}
\end{theorem}

\begin{theorem}\label{Boundcumulants}
The series (\ref{treeformulext}) is absolutely convergent for $\lambda$ small enough, 
its sum is uniformly Borel summable in $\lambda$, and we have:
\begin{equation}\label{decaybound}
\vert  \mathfrak{K}^{2\fm}_{j}(\lambda, x_1, ..., x_{2\fm}) \vert \le (2\fm)!  
\vert \lambda \vert^{\fm-1}  e^{4\pi\fm j}
e^{-cM^j d(x_1, ..., x_{2\fm})}
\end{equation}
where $d(x_1, ..., x_{2\fm})$ is the length of the shortest tree which connects all the
points $x_1, ..., x_{2\fm}$, and $c=\frac{e^{-2\pi}}{2}\sim 0.0157 $.
\end{theorem}

\begin{remark} Our notations differ slightly from \cite{MagRiv}. 
But the number of pairings, i.e., $|{\pi}|$, is the same as in \cite{MagRiv}.
\end{remark}

By combining Lemma \ref{Lem1} and Theorem \ref{Boundcumulants} we arrived at

\begin{lemma}
\bea
\vert  \mathfrak{K}^{2\fm}_{j}(\lambda, x_1, ..., x_{2\fm}) \vert  \le (2\fm)!  
\vert \lambda \vert^{\fm-1}  e^{4\pi\fm j} e^{-cM^j }
\ea
when  $c=\frac{e^{-2\pi}}{2}\sim 0.0157 $.
\end{lemma}
\prf It is obvious since every pair  with $x_{\pi(r,1)}$ and $x_{\pi(r,2)}$ is restricted to $[0,1]^2$. Hence $\vert d(x_1, ..., x_{2\fm})\vert  \le 1$.
\qed

\subsection{From   \texorpdfstring{\cite{MagRiv}}{MagRiv} to \texorpdfstring{\cite{RiZh2}}{RiZh2}}
\label{ss2.4}

Now we are going to perform intermediate field representation from a single slice to multiple slices.
The main problem of the model of \cite{RiZh2} compared to the toy model of \cite{MLVE} is that this action is not positive \cite{nelson}.
Observe that for $\lambda >0$ 
\bee   e^{-\frac{\lambda}{2}\Tr [(\phi^2-3T)^2  -6 T^2]}  \le e^{3 \lambda T^2} , \label{theory1}
\ee
producing Nelson's famous divergent bound \cite{nelson}  as $j_{max} \to \infty$:
\bee  \vert Z^{j_{max}}(\lambda) \vert \le e^{ \lambda O(1) j_{max}^2}   . \label{nelbound}
\ee

Introducing the intermediate field $\sigma$, integrating out the
terms that are quadratic in $\phi(x)$ and using that $\Tr T\sigma = \Tr C \sigma $, \footnote{For simplicity in this subsection we use $C, T$ for $C_{\le j_{max}},T_{\le j_{max}}$.} we get: 
\bea
Z^{j_{max}}(\lambda)&=&\int d\nu(\sigma) e^{\Tr  \bigl(3\lambda  T^2  + 3i \sqrt{\lambda} T
\sigma-\frac{1}{2} \log[1+2i \sqrt{\lambda}C \sigma ]  
\bigr)}, \label{expre0} 
\\ &=&\int d\nu(\sigma) e^{ \Tr  \bigl(3\lambda  T^2  + 2i \sqrt{\lambda} T
\sigma-\frac{1}{2} \log_2[1+2i \sqrt{\lambda}C\sigma ] 
\bigr)}, \label{expre}
\eea
where $d\nu(\sigma)$ is the
ultralocal measure on $\sigma$ with covariance $\delta(x-y)$, and the function
\bee \log_2 (1-x) \equiv x+ \log (1-x) = O(x^2)
\ee 
has to be defined in the operator sense, by the kernel:
\bea [  \log_2 (1+2i \sqrt{\lambda}C\sigma )] (x,y)&=&- \sum_{k=2}^{\infty} \frac{(-2i \sqrt{\lambda})^k}{k}   \int d^2x_1 \cdots  \int d^2x_{k-1} \\
&& \hskip-2cm \bigl[  C (x,x_1) \sigma (x_1)  
C (x_1,x_2)  \cdots  C (x_{k-1},y) \sigma (y)\bigr]. 
\label{loops}
\ea
(Adding a trace on the left-hand side of \eqref{loops} would correspond to multiplying the right-hand side by $\delta(x-y)$ and to integrating over $x$ and $y$.) 
The perturbation theory in terms of $\sigma$ is indexed by intermediate field Feynman graphs whose vertices are the loops in Figure \ref{graphinter6}, obtained by the expansion \eqref{loops} into traces,
and whose $\sigma$-propagators correspond to the former $\phi^4$ \emph{vertices} of ordinary perturbation expansion, hence bear a coupling constant $\lambda$. The loop vertices are themselves cycles of the old $\phi^4$
propagators, which now occur at each \emph{corner} of the loop vertices\footnote{Remark that such intermediate field
Feynman graphs are really combinatorial maps \cite{GRS}. It means that we can define a clockwise
cyclic ordering at each loop vertex. The notion of the \emph{next} intermediate $\sigma$ field (or $\sigma$ half-propagator) at any propagator is then well-defined.}. We 
call these corners $\phi^4$ propagators simply \emph{c-propagators} for short.

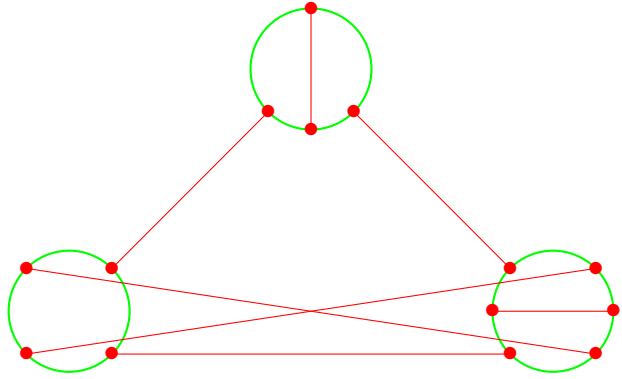
\begin{figure}[!t]
\begin{center}
\begin{tikzpicture}[scale=.8]
\draw [green,thick] (0,0) circle (1) ;
\draw [green,thick] (8,0) circle (1) ;
\draw [green,thick] (4,4) circle (1) ;
\draw [red](0.7071,0.7071) -- (3.2929,3.2929) ;\draw [red] (0.7071,0.7071) node {$\bullet$} ;\draw  [red] (3.2929,3.2929) node {$\bullet$} ;
\draw [red](-0.7071,0.7071) -- (8.7071,-0.7071) ;\draw [red]  (-0.7071,0.7071) node {$\bullet$} ;\draw  [red] (8.7071,-0.7071) node {$\bullet$} ;
\draw [red](4,3) -- (4,5) ;\draw  [red] (4,3) node {$\bullet$} ;\draw  [red] (4,5) node {$\bullet$} ;
\draw [red](4.7071,3.2929) -- (7.2929,0.7071) ;\draw [red] (4.7071,3.2929)  node {$\bullet$} ;\draw  [red] (7.2929,0.7071) node {$\bullet$} ;
\draw [red] (7,0) -- (9,0) ;\draw [red]  (7,0) node {$\bullet$} ;\draw [red]  (9,0) node {$\bullet$} ;
\draw [red](.7071,-.7071) -- (7.2929,-.7071) ;\draw [red]  (.7071,-.7071) node {$\bullet$} ;\draw [red]  (7.2929,-.7071) 
node {$\bullet$} ;
\draw [red](-.7071,-.7071) -- (8.7071,0.7071) ;\draw [red]  (-.7071,-.7071) node {$\bullet$} ;\draw [red]  (8.7071,0.7071) 
node {$\bullet$} ;
\end{tikzpicture}
\end{center}
\caption{An intermediate field graph with 14 vertices and 7 intermediate propagators in red, and 3 cycles in green,
corresponding of ordinary Feynman graph of order $\lambda^7$.}
\label{graphinter6}
\end{figure}

Since the loop interaction in \eqref{expre} is a $\log_2$, the expansion into intermediate field Feynman graphs has no \emph{loop vertex 
of length one}, i.e. no loop vertex with a single corner-propagator. However, renormalization in the intermediate 
field expansion does not reduce to this observation. Wick-ordering sets to zero
all tadpoles in the ordinary perturbative expansion. Accordingly, the counterterm $2i \sqrt{\lambda} T \sigma$ in \eqref{expre} will precisely 
compensate any intermediate field Feynman graph in which a 
$\sigma$-propagator has \emph{length one}, that is, it directly joins the two ends of a c-propagator. Indeed, such a $\sigma$ propagator of length one can be flipped into a loop vertex of length one. This 
is a key difficulty of the LVE in the scalar $\phi^4$ theory, ultimately related to Nelson's bound \eqref{nelbound}.
The intermediate field representation breaks the discrete dualities of the vertex, 
and symmetry breaking, as usual, makes renormalization more difficult \footnote{
This difficulty occurs also in matrix models with quartic interaction, since their vertex has a duality (of order 2 instead of 3).
It does not occur in vector models \cite{MLVE}, nor in tensor models with melonic quartic interactions \cite{DR}, since their vertex has no dualities.}.
 
This difficulty requires us to perform an additional ``slice-testing" expansion.
For any slice $j$, starting with $j_{max}$ and working up to $j=1$, it searches for the presence of \emph{one}
c-propagator or \emph{one} $2i \sqrt{\lambda} T_j$ counter-term of that slice. Then it performs a Wick contraction  of the $\sigma$ 
field \emph{next} to the c-propagator $C_j$, and also of the $\sigma$ field attached to the $T_j$ counter-term, allowing an explicit compensation between tadpole graphs and their counter-terms, ultimately bringing convergent factors $M^{-j/2}$ which will tame Nelson's bound.

However, there is once again a difficulty. 
Testing for the presence of a c-propagator $C_j$ with an interpolation parameter $t_j$ and a Taylor expansion step
does not create just a $C_j$. 
Since $C_j$'s occur within the $\log_2[1+2i \sqrt{\lambda}C\sigma ] $ interaction of \eqref{expre},
derived propagators come equipped with a \emph{resolvent}, which is the operator-valued function of the intermediate field defined by
the same definition of our paper, i.e. $R(\sigma) \equiv [1+2i \sqrt{\lambda}C\sigma ]^{-1}$.

The amplitudes of the combinatorial objects obtained at the end of this slice-testing expansion belong therefore to a new class,
\emph{resolvent amplitudes},  described in \cite{RiZh2}.

The resolvent amplitude $A_G$ of an intermediate-field graph $G$ (with slice attributions $\{j(\ell) \}$) is defined as
\bee  A_G (t, \sigma) =  \prod_{v \in \cV(G)} \bigl[ (-\lambda) \int_{[0,1]^2}  d^2 x_v  \bigr]  \prod_{\ell \in CP(G)} [ R(\sigma)C_{j(\ell)}  ] (x_\ell, x'_\ell ),
\label{ampres}
\ee 
where $\cV(G)$ is the set of $\sigma$-propagators of $G$, 
\hskip-.05cm $CP(G)$ is the set of c-propagators of $G$ and $x_\ell, x'_\ell$ are the positions of the two vertices at the ends of the c-propagator $\ell$. 

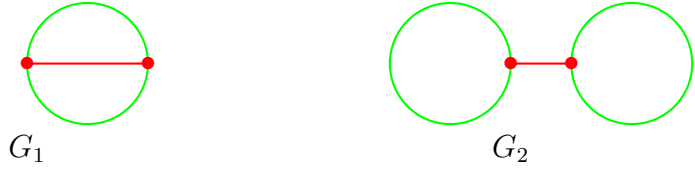
\begin{figure}[!t]
\begin{center}
\begin{tikzpicture}[scale=.8]
\draw [green,thick](0,0) circle (1) ;\draw [green,thick](3,0) circle (1) ;
\draw [red] (1,0) node {$\bullet$} ;
\draw  [red] (2,0) node {$\bullet$} ;\draw [red,thick](1,0) -- (2,0) ;

\draw [green,thick](-6,0) circle (1) ;\draw [red,thick](-7,0) -- (-5,0) ;
\draw  [red] (-7,0) node {$\bullet$} ;\draw  [red] (-5,0) node {$\bullet$} ;

\draw (-7,-1) node[below]{$G_1$} ;\draw (1,-1) node[below]{$G_2$} ;
\end{tikzpicture}
\end{center}
\caption {First-order ($\lambda^1$) intermediate field  \emph{connected graphs (up to symmetry)}, vertices and intermediate propagators pictured in red, cycles pictured in green.}
\label{graphinter7}
\end{figure}

The \emph{renormalized} amplitude of the same graph is the same, but with subtracted tadpole resolvents:
\bea  A^R_G (t, \sigma) &=&  \prod_{v \in \cV(G)} \bigl[ (-\lambda) \int_{[0,1]^2}  d^2 x_v  \bigr]  
\prod_{\genfrac{}{}{0pt}{}{\ell \in CP(G)}{\ell\; tadpole}}  [(R -1)C_{j(\ell)} (\sigma)] (x_\ell, x'_\ell )\nonumber\\
&& \prod_{\genfrac{}{}{0pt}{}{\ell \in CP(G)}{\ell \; not \; tadpole}}  [   R(\sigma)C_{j(\ell)}] (x_\ell, x'_\ell ).
\label{ampresren}
\ea
Hence, in renormalized resolvent amplitudes, c-propagators are not ordinary $C's$ but either
$(R-1)C$ or $RC$ depending whether the c-propagator is a tadpole or not. They correspond, therefore, when expanding the $R$
or $R-1$ factors and performing the $\sigma$ Wick contractions, to \emph{infinite series}
of ordinary $\phi^4$ graphs, but with the particularity that the initial propagators $C$'s of the renormalized resolvent graph
cannot be tadpoles. For instance, at order 1 and 2 there are two connected resolvent graphs
associated with the two intermediate graphs of Figure \ref{graphinter7} and Figure \ref{graphinter8}. One graph (that has the label $G_1$)  has a tadpole, and for that tadpole simply replace the c-propagator with the $(R-1)C$ factor. 

\begin{figure}[!t]
\begin{center}
\begin{tikzpicture}[scale=.8]

\draw [green,thick](0,0) circle (1) ;
\draw [green,thick](3,0) circle (1) ;
\draw [red] (1,0) node {$\bullet$} ;
\draw  [red] (2,0) node {$\bullet$} ;
\draw [red,thick](0,-1) -- (0,1) ;
\draw [red,thick](1,0) -- (2,0) ;\draw  [red] (0,1) node {$\bullet$} ;\draw  [red] (0,-1) node {$\bullet$} ;

\draw [green,thick](8,0) circle (1) ;
\draw [green,thick](11,0) circle (1) ;
\draw [green,thick](14,0) circle (1) ;

\draw [red,thick](9,0) -- (10,0) ;\draw [red,thick](12,0) -- (13,0) ;
\draw  [red] (9,0) node {$\bullet$} ;\draw  [red] (10,0) node {$\bullet$} ;
\draw  [red] (12,0) node {$\bullet$} ;\draw  [red] (13,0) node {$\bullet$} ;
\draw (1.5,-1) node[below]{$G_1$} ;
\draw (9.5,-1) node[below]{$G_2$} ;
\end{tikzpicture}
\end{center}
\caption {Second-order ($\lambda^2$) intermediate field \emph{connected graphs (up to symmetry)}, vertices and intermediate propagators pictured in red, cycles pictured in green.}
\label{graphinter8}
\end{figure}
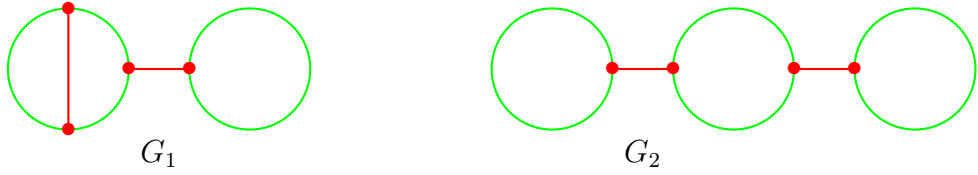

Consider any \emph{slice-subset} $\mJ : \{j_1, \cdots j_p\} \subset \cS = [1, \cdots, j_{max}]$.\footnote{We really need the parameter $J$ of the first part of \cite{RiZh2}, therefore the parameter $J$ of the second part of \cite{RiZh2} is changed into $\mJ$.}
A $\mJ$-resolvent graph is defined as a resolvent graph which exactly
$p$ c-propagators must bear marks $\{ j_1, \cdots, j_p\}$. Finally a 
$\mJ$-resolvent graph is called \emph{minimal} if every connected component of the graph bears at least a mark,
and the total perturbative order of the graph, i.e. the total number of $\sigma$-propagators, is at most $\vert \mJ \vert$. The set of minimal
$\mJ$-resolvent graphs is noted $\cG (\mJ)$, and we denote $\cG = \cup_{\mJ} \cG (\mJ)$. By convention, we could say that
$\mJ = \emptyset$ is allowed, resulting in a single ``empty graph" in $ \cG (\emptyset)$ with no propagator. It will correspond to the free theory, hence to the term 1 in the expansion of $  Z^{j_{max}}(\lambda)$ below.
A factorization expansion quite similar to that of \cite{MLVE} can be performed. The marked propagators provide the good factors which ultimately pay for the Nelson bound and all combinatorics.

\subsection{Slice-testing Expansion} \label{slitest}

We introduce inductively interpolation parameters $t_j \in [0,1]$ for the $j$-th scale of the propagators. 

We shall write simply $t$ for the family $\{t_j\}, 1\le j \le j_{max}$. This means that we write
\bee  C(t)  = \sum_{j=1}^{j_{max} } t_j C_j , \quad  T(t)  = \sum_{j=1}^{j_{max} }  t_j T_j ,\;\;  T_j = C_j (x, x),
\ee
\bee V (t)   = \Tr  \bigl( 3\lambda  T^2(t) + 2i \sqrt{\lambda} T(t) 
\sigma-\frac{1}{2} \log_2[1+2i \sqrt{\lambda}C(t)\sigma ]  \bigr) . \label{simpleslic}
\ee

We perform a single first-order Taylor expansion step in each $t_j$ between 0 and 1. It results in
\bee  Z^{j_{max}}(\lambda)=   \sum_{\mJ \subset \cS} \int d\nu  (\sigma)  \prod_{j \in \mJ}  \int_0^1 dt_j\frac{d}{dt_j} e^{ V (t) } \vert_{t_j = 0 \; {\rm for} \; j \; \not\in \mJ} . \label{testing}
\ee

Each $\frac{d}{dt_j}$ hits either a propagator or a tadpole, resulting in a well-defined $T_j$ or a well-defined $C_j \sigma$ brought down from the exponential.
Using that $(\log_2 (1+x) )' = (1+x)^{-1} -1 = - x/(1+x)$, the $T_j$'s  comes equipped with a $T(t) $ or $ \sqrt{\lambda} \sigma $, and the $C_j$'s comes equipped with an $R(t)  -1$, where 
\bee  R(t) = [1+2i \sqrt{\lambda}C(t)\sigma ]^{-1} .
\ee

The following Lemma is proved in \cite{RiZh2}:
\begin{lemma} \label{testinglemma}
The general term of the testing expansion is 
\bee  Z^{j_{max}}(\lambda)=   \sum_{\mJ \subset \cS}
\sum_{G \subset \cG (\mJ)}  c_G
\int d\nu  (\sigma)     \prod_{j \in \mJ (G)}  
\int_0^1 dt_j  \; \bigl[ e^{V (t) }  A^R_G (t, \sigma)   \bigr]_{t_j = 0 \; {\rm for} \; j \; \not\in \mJ (G) } \label{testing1}
\ee 
where the sum over $\cG (\mJ)$ runs over a set of minimal (vacuum) resolvent graphs (connected or not). 
The renormalized amplitudes $A^R_G$
are defined by \eqref{ampresren}, where the index $j(\ell)$ specifies the markings, that is restricts the c-propagator $\ell$ to 
be $C_j$ if that propagator bears the mark $j$. All marked propagators which belong to a short cycle (i.e. potential tadpole)
are \emph{renormalized}, hence accompanied by an $R-1$ resolvent factor, and not an $R$ factor. 
The c-propagators which do not bear any mark are equal to $C(t)$. 
\end{lemma}

See \cite{RiZh2} for further details, in particular the proof of the above-mentioned Lemma.

\subsection{Results of  \texorpdfstring{\cite{RiZh2}}{RiZh2}}

We follow the same steps as in \cite{MLVE} 
and obtain a two-level-jungle formula \cite{AR1}. It writes
\bee
Z^{j_{max}}(\lambda)  =  \sum_{n=0}^\infty \frac{1}{n!}  \sum_{\cJ} \;
 \;  \int dw_\cJ  \;  \int d\nu_{ \cJ}  
\quad   \partial_\cJ   \Big[ \prod_{\cB} \prod_{a\in \cB}      W_{a}   (   \sigma_a , \chi^{ \cB } , \bar \chi^{\cB}  )
  \Big] \; ,
\ee
where
\begin{itemize}

\item the sum over $\cJ$ runs over all two-level jungles,
hence over all ordered pairs $\cJ $ of two (each possibly empty) 
disjoint forests on $\cW$, such that $\cF_B$ is a forest, $\cF_F$ is a forest and
$\cJ = \cF_B \cup \cF_F $ is still a forest on $\cW$. 
The forests $\cF_B$ and $\cF_F$ are the Bosonic and Fermionic components of $\cJ$.
 
\item  $\int dw_\cJ$ means integration from 0 to 1 over parameters $w_\ell$, one for each edge $\ell \in \cJ$, namely
$\int dw_{\ \cJ } = \prod_{\ell\in  \cJ}  \int_0^1 dw_\ell  $.
There is no integration for the empty forest since, by convention, an empty product is 1. A generic integration point $w_\cJ$
is therefore made of $\vert  \cJ \vert$ parameters $w_\ell \in [0,1]$, one for each $\ell \in  \cJ$.

\item 
\bea \partial_\cJ  &=& \prod_{\genfrac{}{}{0pt}{}{\ell_B \in \cF_B}{\ell_B=(a,b)}} \Bigl(
\frac{\partial}{\partial \sigma_a}\frac{\partial}{\partial \sigma_b} \Bigr)
\prod_{\genfrac{}{}{0pt}{}{\ell_F \in \cF_F}{\ell_F=(d,e) } } \sum_{j_{\ell_F} =1}^{j_{max}} \Big(
   \frac{\partial}{\partial \bar \chi^{\cB(d)}_{j_{\ell_F} } }\frac{\partial}{\partial \chi^{\cB(e)}_{j_{\ell_F} } } + 
    \frac{\partial}{\partial \bar \chi^{ \cB( e) }_{j_{\ell_F} } } \frac{\partial}{\partial \chi^{\cB(d) }_{j_{\ell_F}  } }
   \Big) \nonumber\\&&
\ea
where $ \cB(d)$ denotes the Bosonic block to which the vertex $d$ belongs. 

\item The measure $d\nu_{\cJ}$ has covariance $ X (w_{\ell_B}) $ on Bosonic variables and $ Y (w_{\ell_F}) \otimes \mathbb{I}_\cS  $  
on Fermionic variables, hence
\bea\label{EQ32}
\int d\nu_{\cJ} F &=& \biggl[
e^{\frac{1}{2} \sum_{a,b=1}^n X_{ab}(w_{\ell_B })  \frac{\partial}{\partial \sigma_a}\frac{\partial}{\partial \sigma_b} }
 \\ \nonumber&& \hskip1cm e^{ \sum_{\cB,\cB'} Y_{\cB\cB'}(w_{\ell_F})  \sum_{j_{\cB, \cB'} \in \cS}
   \frac{\partial}{\partial \bar \chi_{j_{\cB, \cB'}}^{\cB} } \frac{\partial}{\partial \chi_{j_{\cB, \cB'}}^{\cB'} }  }  F  \biggr]_{\sigma = \bar\chi =\chi =0}\; .
\ea

\item  $X_{ab} (w_{\ell_B} )$  is the infimum of the $w_{\ell_B}$ parameters for all the Bosonic edges $\ell_B$
in the unique path $P^{\cF_B}_{a \to b}$ from $a$ to $b$ in $\cF_B$. This infimum is set to zero if such a path does not exist and 
to $1$ if $a=b$. 

\item  $Y_{\cB\cB'}(w_{\ell_F})$  is the infimum of the $w_{\ell_F}$ parameters for all the Fermionic
edges $\ell_F$ in any of the paths $P^{\cF_B \cup \cF_F}_{a\to b}$ from some vertex $a\in \cB$ to some vertex $b\in \cB'$. 
This infimum is set to $0$ if there are no such paths, and to $1$ if such paths exist but do not contain any Fermionic edges.

\end{itemize}

Remember that a main property of the forest formula is that the symmetric $n$ by $n$ matrix $X_{ab}(w_{\ell_B})$ 
is positive for any value of $w_\cJ$, hence the Gaussian measure $d\nu_{\cJ} $ is well-defined. 

Since the slice assignments, the fields, the measure and the integrand are now 
factorized over the connected components of $\cJ$, the logarithm of $Z$ is easily computed as exactly the same sum but restricted 
to two-level spanning trees:
\bea \label{treerep1}  
\log Z^{j_{max}}(\lambda)\!&=&\!  \sum_{n=1}^\infty \frac{1}{n!}  \sum_{\bar \cJ }   \int dw_{\bar\cJ }   \int d\nu_{\bar\cJ } 
\,  \partial_{\bar\cJ }   \Big[ \prod_{\cB} \prod_{a\in \cB}   \Bigl(     W_{a}   (   \sigma_a , \chi^{ \cB } , \bar \chi^{\cB}  )\Bigr)    \Big] , 
\nonumber\\&&
\ea
where the sum is the same but conditioned on $\bar \cJ = \cF_B \cup \cF_F$ being a \emph{spanning tree} on $\cW= [1, \cdots , n]$.
The main result is the convergence of this representation uniformly in $j_{max}$ for $\lambda$ in a certain domain, 
allowing the ultraviolet limit of the theory to be performed. 

More precisely in \cite{RiZh2} the following Theorem is proven:

\begin{theorem} \label{thetheorem} Fix  $\rho >0$ small enough.
The series \eqref{treerep1}
is absolutely convergent, uniformly in $j_{max}$, for $\lambda$ in the small open cardioid domain 
defined by $\vert \lambda \vert < \rho \cos^2 [({\rm Arg} \; \lambda )/2]$
(see Figure \ref{cardio}).
Its ultraviolet limit $\log Z (\lambda) = \lim_{j_{max}  \to \infty}  \log Z^{j_{max}}(\lambda)$ is therefore well-defined and
analytic in that cardioid domain; furthermore it is the Borel sum
of its perturbative series in powers of $\lambda$. 
\end{theorem}

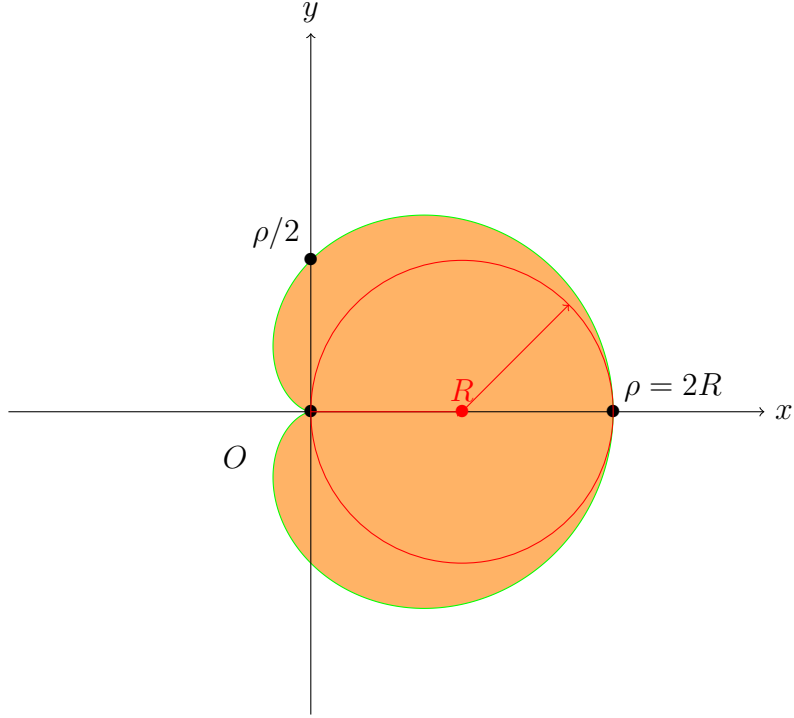
\begin{figure}[!t]
\begin{center}
\begin{tikzpicture}[scale=2]
\tikzstyle{red2}=[thick,  fill=red!40]
\draw (-0.5,-0.3) node {$O$} ;
\draw (2,0) node[above right]{$\rho=2R$} ;
\draw (0,1) node[above left]{$\rho/2$} ;
\filldraw [draw=green,fill=orange!60,domain=-180:180,samples=400,scale=.5,variable=\t] plot (\t:  2+ 2*cos \t);
\draw (-0,0) node {$\bullet$} ;
\draw (2,0) node {$\bullet$} ;
\draw (0,1) node {$\bullet$} ;
\draw[->] (-2,0) -- (3,0);
\draw (3,0) node[right] {$x$};
\draw [->] (0,-2) -- (0,2.5);
\draw (0,2.5) node[above] {$y$}; 
\draw[red] (1,0) circle (1) ;
\draw[red] (1,0) node {$\bullet$} ;
\draw[red] (1,0) node[above]{$R$} ;
\draw[red,->] (0,0) -- (1,0);
\draw[red,->] (1,0) -- (1.7071,0.7071);
\end{tikzpicture}
\end{center}
\caption{Open cardioid domain pictured in orange, frontier pictured in green. In red, $D_R$ for $2R=\rho$, see Appendix 
B.}
\label{cardio}
\end{figure}

See \cite{RiZh2} for further details, in particular the proof of the above-mentioned Theorem.
 
\section{The MLVE with Cumulants}
\label{mainresult}

\begin{definition}\label{def3}
Let us define a \emph{resolvent} by a multi-slice operator 
\bea
C_{\le j_{max}}(\sigma_r, x, y ), \quad C_{\le j_{max}}(\sigma_r, x, y) = D_{\le j_{max} } \frac{1}{1+i H_{\le j_{max} }  } D_{\le j_{max} }  ,
\label{eq33}
\ea
with the same notations than in Section \ref{SinSlice}. Let us define  $\cV^r$ by 
\bea
\cV^r&:=&  - \frac 12 \Tr\log_2 (1+iH_{\le j_{max} } ).  \label{EQ33}
\ea
\end{definition}

In the intermediate field representation, the cumulant of even order, by analogy with \cite{MagRiv}, Theorem 4.1, which is in this paper Theorem \ref{cumulants}, is
\bea 
\mathfrak{K}^{2\fm}_{j_{\max}}(\lambda, x_1, ..., x_{2\fm})
&=& \sum_{n=1}^{\infty}\frac{1}{n!} \sum_\cJ \Big\{ \prod_{\ell\in \cJ}   
\big[ \int_0^1 dw_\ell \int d^2 x_\ell d^2 y_\ell \big] \Big\} \nonumber \\
&&\hskip-4.5cm \bigg\{ \prod_{r=1}^{\fm} \int  d\nu_\cJ (\{\sigma_v\}, \{\sigma_r\}, \{ w \}) 
 \prod_{\ell\in \cJ} \big[ \delta (x_\ell - y_\ell) 
 \frac{\delta}{\delta \sigma_{v(\ell)}(x_\ell)}\frac{\delta}{\delta \sigma_{v'(\ell)}(y_\ell)} 
 \big]  \nonumber \\
 && \prod_v \cV^r_v \sum_{\pi}  C_{\le j_{max}}
 (\sigma_{r}, x_{\pi(r,1)}, x_{\pi(r,2)})\bigg\}\; , \label{eq32}
\ea
with the same notations than in Theorem \ref{cumulants} and Definition \ref{def3}.

\subsection{MLVE Amplitudes}\footnote{This section is essentially a condensed introduction of \cite{GuKra}.}

The perturbation theory in terms of $\sigma$ is indexed by intermediate field Feynman graphs (see Figure \ref{graphinter6}) whose vertices are the loops obtained by the expansion \eqref{loops} into traces,
and whose $\sigma$-propagators, represented by red lines in Figure \ref{graphinter6}, correspond to the former $\phi^4$ \emph{vertices} of ordinary perturbation expansion, hence bear a coupling constant $\lambda$. The loop vertices are themselves cycles of the old $\phi^4$
propagators, which now occur at each \emph{corner} of the loop vertices. Remark that such intermediate field
Feynman graphs are really combinatorial maps \cite{GRS}. This means that we can define a clockwise
cyclic ordering at each loop vertex. The notion of the \emph{next} intermediate $\sigma$ field (or $\sigma$ half-propagator) at any propagator is then well-defined. We call these corner $\phi^4$ propagators simply \emph{c-propagators} for short.
The perturbative order of an intermediate graph is the total number of $\sigma$-propagators. 
Adding sources would introduce, in addition to the loop vertices, resolvents, which can be considered as ciliated loop vertices
\cite{Gurau:2013pca}.

For the task of computing $\mathfrak{K}^{2\fm}_{j_{max}}(\lambda,x_1,\cdots,x_{2\fm}) $, we introduce multi-scale combinatorial maps, a refinement of the usual Feynman graphs.
First, we define a combinatorial map with the same definition as in \cite{GuKra}. A combinatorial map is a graph with a distinguished cyclic ordering of the half edges incident at each vertex. 

Combinatorial maps are conveniently represented as \emph{ribbon graphs} whose vertices are disks and whose edges are ribbons (allowing one to encode the ordering of the half edges incident at a vertex). When applied to cumulants, it is based on combinatorial maps with cilia. 
A \emph{cilium} is a half edge hooked to a vertex.

We denote $\fm(G)$, $v(G)$, $\cV(G)$, $e(G)$, $f(G)$ and $c(G)$ the cilia,
vertices, c-propagators, edges, faces and \emph{corners} of $G$, and their \emph{numbers} $|\fm(G)|$, $|v(G)|$, $|\cV(G)|$, $|e(G)|$, $|f(G)|$,  $|c(G)|$. A corner of $G$ is a pair of consecutive half-edges attached to the same vertex. 

The edges of $G$ not in $\cJ$ are called \emph{loop edges} and we denote $ e(G,\cJ)= e(G)- e(\cJ)$ the set of 
loop edges. It depends on $\cJ$; for instance, there is no way to picture them in Figure  \ref{graphinter6}.

The faces of $G$ are partitioned between the faces which do not contain any cilium (which we sometimes call internal faces) 
and the ones which contain at least one cilium, which we call \emph{broken faces}.  We denote $b(G)$ as the set of broken faces of $G$. 
Each broken face corresponds to a puncture in the Riemann surface in which $G$ is embedded; see Figure \ref{Fig3a}.
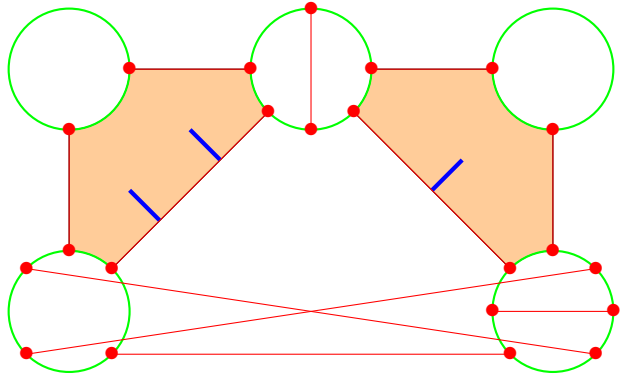
\begin{figure}[!htb]\centering

\begin{tikzpicture}[scale=.8]

\filldraw[draw=black,fill=orange!40] plot (0,1) -- (0,3) -- plot (0,3) arc (-90:0:1)  -- plot (1,4) -- (3,4) -- plot (3,4) arc (-180:-135:1) 
 -- (0.7071,0.7071) -- plot (0.7071,0.7071) arc (45:90:1) -- cycle; 
 \filldraw[draw=black,fill=orange!40] 
plot (5,4) -- (7,4)  -- plot (7,4) arc (-180:-90:1)  -- plot (8,3) -- (8,1) -- plot (8,1) arc (90:135:1)  
 --  (7.2929,0.7071) -- plot  (4.7071,3.2929) arc (-45:0:1)  -- cycle; 
\draw [green,thick] (0,0) circle (1) ;
\draw [green,thick] (8,0) circle (1) ;
\draw [green,thick] (4,4) circle (1) ;
\draw [green,thick] (0,4) circle (1) ;
\draw [green,thick] (8,4) circle (1) ;
\draw [red](1,4) -- (3,4) ;\draw  [red] (1,4) node {$\bullet$} ;\draw  [red] (3,4) node {$\bullet$} ;
\draw [red](5,4) -- (7,4) ;\draw  [red] (5,4) node {$\bullet$} ;\draw  [red] (7,4) node {$\bullet$} ;
\draw [red](0,1) -- (0,3) ;\draw  [red] (0,1) node {$\bullet$} ;\draw  [red] (0,3) node {$\bullet$} ;
\draw [red](8,1) -- (8,3) ;\draw  [red] (8,1) node {$\bullet$} ;\draw  [red] (8,3) node {$\bullet$} ;

\draw [blue,ultra thick](1,2) -- (1.5,1.5) ;
\draw [blue,ultra thick](2,3) -- (2.5,2.5) ;
\draw [blue,ultra thick](6,2) -- (6.5,2.5) ;

\draw [red](0.7071,0.7071) -- (3.2929,3.2929) ;\draw [red] (0.7071,0.7071) node {$\bullet$} ;\draw  [red] (3.2929,3.2929) node {$\bullet$} ;
\draw [red](-0.7071,0.7071) -- (8.7071,-0.7071) ;\draw [red]  (-0.7071,0.7071) node {$\bullet$} ;\draw  [red] (8.7071,-0.7071) node {$\bullet$} ;
\draw [red](4,3) -- (4,5) ;\draw  [red] (4,3) node {$\bullet$} ;\draw  [red] (4,5) node {$\bullet$} ;
\draw [red](4.7071,3.2929) -- (7.2929,0.7071) ;\draw [red] (4.7071,3.2929)  node {$\bullet$} ;\draw  [red] (7.2929,0.7071) node {$\bullet$} ;
\draw [red] (7,0) -- (9,0) ;\draw [red]  (7,0) node {$\bullet$} ;\draw [red]  (9,0) node {$\bullet$} ;
\draw [red](.7071,-.7071) -- (7.2929,-.7071) ;\draw [red]  (.7071,-.7071) node {$\bullet$} ;\draw [red]  (7.2929,-.7071) 
node {$\bullet$} ;
\draw [red](-.7071,-.7071) -- (8.7071,0.7071) ;\draw [red]  (-.7071,-.7071) node {$\bullet$} ;\draw [red]  (8.7071,0.7071) 
node {$\bullet$} ;
\end{tikzpicture}
\caption{A LVE graph with five cycles in green, vertices and intermediate c-propagators in red,
three cilia in blue and two broken faces in orange, corresponding to a Feynman graph of order 
$\lambda^{10}$.} 
\label{Fig3a}
\end{figure}

The following Lemma is trivial if one is familiar of \cite{RiZh1,GuKra} (Lemma 4 of \cite{GuKra}):
\begin{lemma}
\label{Heppsectors}
For any vertex-labeled graph $G$:
\begin{equation}
\sum_{\genfrac{}{}{0pt}{}{\cJ\subset G}{ \cJ\, \text{ spanning tree} } }
\prod_{\ell\in \cV(\cJ)}\int_0^1 dw_{\ell}\, \left( \prod_{\ell=(i,j)\in e(G,\cJ)}
\inf_{\ell'\in P^\cJ_{i\leftrightarrow j}}w_{\ell'} \right) =1 \; .
\end{equation}
\end{lemma}

Next we bound the number of LVE graphs with a given number of vertices, loop edges, and cilia.

\begin{lemma}[Counting LVE trees]
\label{coutingtrees:lem}
The number of LVE trees with $v$ edges and  ${\mathcal K}$ cilia is
\bea
\label{EQ45.0}
{\cal N}(v,{\mathcal K})&=&\frac{(2v+{\mathcal K}-1)!\,(v+1)!}{(v+{\mathcal K})!\,(v+1-{\mathcal K})!\,{\mathcal K}!}\\
\label{EQ47.0}
&\le& 2^{3v+ {\mathcal K}}\,(v-1)!\,
\ea
\end{lemma}
\prf The first part, namely \eqref{EQ45.0}, is already proved in \cite{GuKra}, Lemma 3.
Using the binomial bound
\bea
\frac{(2v+{\mathcal K}-1)!}{(v+{\mathcal K}!)(v-1)!} \le 2^{2v+{\mathcal K}-1},\quad
\frac{  (v+1)!}{  
   (v+1-{\mathcal K})!\,    {\mathcal K}!   }
\le 2^{ v+ 1 } ,
\ea
\eqref{EQ47.0} is proved.
\qed
\begin{lemma}[Counting LVE graphs]
The number of LVE graphs with $v+1$ vertices, $\ell$ loop edges and ${\mathcal K}$ cilia reads 
\bea\label{EQ55}
{\cal N}(v,\ell,{\mathcal K})&=&\frac{(2 v+2 \ell +{\mathcal K}-1)!(v+1)!}{(v+{\mathcal K})!2^{\ell}{\mathcal K}! 
  (v+1-{\mathcal K})!  } \; .
 \\\label{EQ56}&\le& 2^{3v   +\ell  +{\mathcal K}}          (v + 2\ell -1)! 
\ea
\end{lemma}
\prf
The first part of this lemma is proved in \cite{GuKra}, Lemma 10.
The second part is our own, using the bound 
\bea
\frac{(2v+ 2\ell +{\mathcal K} -1)!}{(  v +{\mathcal K}  )!  }
\le 2^{2v+2\ell   +{\mathcal K}   -1}   (v + 2\ell -1)!
\ea
and
\bea
\frac{  (v+1)!}{  (v+1-{\mathcal K})!\,    {\mathcal K}!     }
\le 2^{ v+ 1 } .
\ea
\qed

\subsection{Results}

Let us recall that in our article we worry only about cumulants of even order. This means that $|\fm(G)|$ has to be some even number, i.e. $|\fm(G)|=2\fm$.
Now we extend the amplitudes of \cite{RiZh2} with cilia.

\medskip
\begin{definition}\label{def4}
The resolvent amplitude $A_{G,\fm(G)}$ of an intermediate-field graph $G$ \emph{with slice attributions $\{j(\ell) \}$} is defined as
\bea  A_{G,\fm(G)} (t, \sigma) =  \prod_{v\in \cV(G)} \Big\{ (-\lambda) \int_{[0,1]^2}  d^2 x_v  \
\prod_{\ell \in CP(G)} [ R(\sigma)C_{j(\ell)}  ] (x_\ell, x'_\ell ) \Big\}
\label{ampres2}
\ea 
where $\cV(G)$ is the set of $\sigma$-propagators of $G$ and $CP(G)$ is the set of c-propagators of $G$. 
For  $R(\sigma)$ see subsection \ref{ss2.4} and in particular \eqref{ampres}. $x_\ell, x'_\ell$ are the positions of the two vertices
at the ends of the c-propagator $\ell$. 
Let us define   
\bea\label{E47}
{\mJ^\fm}=\{j^1,\cdots , j^\fm \}
\ea
for the \emph{vector} $\{j^1,\cdots , j^\fm \}$,
and explicitely mark  the dependence in $\lambda$ in $\cV^r$ by defining:
\bea 
\prod_v \cV^r_v := (-\lambda)^{v} \prod_v \cD^r_v .
\ea
\end{definition}
We expend $\mathfrak{K}^{2\fm}_{j_{\max}}(\lambda, x_1, ..., x_{2\fm})$ into resolvent un-renormalized amplitudes 
\bea
\mathfrak{K}^{2\fm}_{j_{\max}}(\lambda, x_1, ..., x_{2\fm}) &=&
 \sum_{A_{G,\fm(G)} }  A_{G,\fm(G)}(t, \sigma) \\ && \hskip-3.5cm =
 \sum_{n=1}^{\infty}\frac{1}{n!}   \sum_{\cJ}\prod_{\ell\in \cJ}    
\big[ \int_0^1 dw_\ell \int d^2 x_\ell d^2 y_\ell \big] (-\lambda )^v \prod_v \cD_v  \nonumber \\
&&\hskip-3.5cm\int  d\nu_\cJ (\{\sigma_v\}, \{\sigma_r\}, \{ w \}) 
\prod_{\ell\in \cJ} \big[ \delta (x_\ell - y_\ell) 
\frac{\delta}{\delta \sigma_{v(\ell)}(x_\ell)}\frac{\delta}{\delta \sigma_{v'(\ell)}(y_\ell)}  \big]  
\nonumber\\&& \hskip-1.5cm  \prod_{r=1}^{\fm}
 \bigg[  \sum_{\pi} 
\sum_{\mJ^{\fm}=\{1,\cdots,1\}}^{{\mJ^\fm}=\{j_{max},\cdots,j_{max}\}}
C_j (\sigma_{r}, x_{\pi(r,1)}, x_{\pi(r,2)})\bigg]\;  ,
\ea
where the notations are the same than in \eqref{treeformulext} and Definition \ref{def4}.
But the expansion over unrenormalized amplitudes is just a warm-up exercise and is no help to prove Theorem  \ref{THM1}. 
Now we expand $\mathfrak{K}^{2\fm}_{j_{\max}}(\lambda, x_1, ..., x_{2\fm})$ into \emph{renormalized} amplitudes, not an easy task because we have to accommodate Nelson's tadpoles.

\begin{definition}
The \emph{renormalized} amplitude of the same graph is the same, but with subtracted tadpole resolvents:
\bea \label{Mampresren}
A^R_{G,\fm(G)} (t, \sigma) &=&  \prod_{v \in \cV(G)} \bigl[ (-\lambda) \int_{[0,1]^2}  d^2 x_v  \bigr] 
\\&&\nonumber
\prod_{\genfrac{}{}{0pt}{}{\ell \in CP(G)}{\ell\; tadpole}}  
[(R -1)C_{j(\ell)} (\sigma)] (x_\ell, x'_\ell ) \prod_{\genfrac{}{}{0pt}{}{\ell \in CP(G)}{\ell \; not \; tadpole}}  [   R(\sigma)C_{j(\ell)}] (x_\ell, x'_\ell ).
\ea 
\end{definition}

Next we have to compute the cumulants with this new definition.
\bea
\mathfrak{K}^{2\fm}_{j_{\max}}(\lambda, x_1, ..., x_{2\fm}) &=&  
\sum_{A^R_{G,\fm(G)} } A^R_{G,\fm(G)}(t, \sigma)
\\&=&\Big[{\frac{\partial^{2\fm}}{\partial J(x_1)\cdots\partial J(x_{2\fm})}}\log{ Z(\lambda,j_{\max},J)} \Big]_{\{J\}=0}.
\label{EQ40}
\ea
It seems that we have to compute $\log{ Z(\lambda,j_{\max},J)}$.
Fortunately remark that we have to compute only $\mathfrak{K}^{2\fm}_{j_{\max}}(\lambda, x_1, ..., x_{2\fm})$ 
for $1\le \fm \le \fm_{max}$.
\bea\label{EQ43}
\mathfrak{K}^{2\fm}_{j_{\max}}(\lambda, x_1, ..., x_{2\fm}) &=&
\sum_{n=1}^{\infty}\frac{1}{n!}   (-\lambda)^v \prod_v \cD^r_v 
\bigg\{
\sum_{\bar \cJ_\fm}  \int dw_{\bar \cJ_\fm}  \int d\nu_{\bar \cJ_\fm} \partial_{\bar \cJ_\fm}
\\ &&  \nonumber\hskip-4cm 
\bigg[ \prod_{\cB} \prod_{a\in \cB}   \Bigl(     W_{a}   ( \sigma_a , \chi^{ \cB } , \bar \chi^{\cB}  )\Bigr)  \bigg] 
\prod_{r=1}^{\fm} \sum_{\pi} \sum_{\mJ^{\fm}=\{1,\cdots,1\}}^{{\mJ^\fm}=\{j_{max},\cdots,j_{max}\}}C_j (\sigma_{r}, x_{\pi(r,1)}, x_{\pi(r,2)})  \bigg\}\;  .
\ea
with the same definitions than   \eqref{eq33}-\eqref{EQ33}-\eqref{eq32}-\eqref{E47}, except that 
\begin{itemize}
\item \emph{the external faces are renormalized}, 
\item $d\nu_{\bar \cJ_\fm} (\{\sigma_v\}, \{\sigma_r\}, \{ w \}) $  has covariance  
\bea
<\sigma_\alpha,\sigma_{\alpha'}>=
\delta (x-y) w^{\bar \cJ_\fm}  (\alpha, \alpha', \{ w\})
\ea where $w^{\bar \cJ_\fm}  (\alpha, \alpha', \{ w\})$ is 1 if $\alpha=\alpha'$ \emph{where $\alpha, \alpha'\in \{v\},  \{\sigma_r\}$}
and is the infimum of the $w_\ell$ for $\ell$ running over the unique path from $\alpha$ to $\alpha'$ in ${\bar \cJ_\fm} $
if $\alpha\ne \alpha'$.
\end{itemize}
\medskip

Our main result extend those of \cite{GuKra,MLVE,RiZh2} to the \emph{cumulants of order $2\fm$ with $1\le\fm
\le \fm_{\max}
$}.
\begin{theorem}
\label{THM1} 
Fix $1\le \fm\le \fm_{\max}$. Let  $\lambda\in \mathbb{C}$ be in the domain 
\bea \label{T0}
|\lambda | \le \rho \cos^2 [({\rm Arg} \; \lambda )/2]
\ea
for $\rho >0$ being small enough depending on $ \fm_{\max}\ge 1 $. 
The series \eqref{cum1} is convergent in that domain, i.e.
\bea \label{T1}
| \mathfrak{K}^{2\fm}_{j_{\max}}(\lambda,x_1, ..., x_{2\fm}) | \le  \rho 
 \cos^2 [({\rm Arg} \; \lambda )/2]  .
\ea
Under the same conditions,  $\mathfrak{K}^{2\fm}_{j_{\max}}(\lambda,x_1, ..., x_{2\fm})$ 
is convergent \emph{uniformly in $j_{\max}$} and its ultraviolet limit 
\bea \label{T2}
\lim_{j_{max}  \to \infty} \mathfrak{K}^{2\fm}_{j_{\max}}(\lambda,x_1, ..., x_{2\fm})=\mathfrak{K}^{2\fm}(\lambda,x_1, ..., x_{2\fm})
\ea is therefore well-defined and analytic in that domain.
Furthermore, it is the Borel sum of its perturbative series in powers of $\lambda$.
\end{theorem}

The proof of this theorem is given in Section~\ref{sec4}. But
there is a complication that is not in \cite{GuKra,Riv1,MLVE,RiZh2}: the 
left parts of the cilia are joined by external faces $ef(G)$, and their number is $|ef(G)|$. Similarly, the right parts of the cilia are joined by external faces $ef'(G)$. To solve this difficulty, remark that their number must be equal:  $|ef(G)|=|ef'(G)|$.
In Theorem \ref{THM1} the cumulants are of order $2\fm$ with $1\le\fm\le \fm_{\max}$, and that is indeed a crucial point.
It means that the combinatorics of external faces is bounded. In fact we have 
\bea\label{E39}
C^{ef(G)}_{2\fm} =\frac{[2\fm]!}{[ef(G)] ![ef'(G)] !}  \le 2^{2\fm} =4^\fm \le 4^{\fm_{\max}} .
\ea

\section{Proof of Theorem \ref{THM1}}
\label{sec4}
We assume for now that our reader is familiar with \cite{RiZh2}.
Here we go.

\begin{definition}\label{def7}
\bea\label{eq61}
\hskip-.7cm\mathfrak{K}^{2\fm}_{j_{\max}}(\lambda,x_1, ..., x_{2\fm}) &=&   \sum_{\bar \cJ_\fm}  \int dw_{\bar \cJ_\fm} 
 \int d\nu_{\bar \cJ_\fm} \partial_{\bar \cJ_\fm}  \nonumber \\&&
\mathfrak{L}^{2\fm}_{j_{\max}}(\lambda,\sigma) \times \mathfrak{M}^{2\fm}_{j_{\max}}(x_1, ..., x_{2\fm}),
\\
\hskip-.7cm\mathfrak{L}_{j_{\max}}(\lambda,\sigma)&:=&\sum_{n=1}^{\infty}\frac{1}{n!}   (-\lambda)^{v}\prod_v \cD_v 
\prod_{\cB} \prod_{a\in \cB}   \Bigl(     W_{a}   ( \sigma_a , \chi^{ \cB } , \bar \chi^{\cB}  )\Bigr),\label{eq62}
\\
\hskip-.7cm\mathfrak{M}^{2\fm}_{j_{\max}}(x_1, ..., x_{2\fm})&:=&
\bigg[  \prod_{r=1}^{\fm}\sum_{\pi} \sum_{\mJ^{\fm}=\{1,\cdots,1\}}^{{\mJ^\fm}
=\{j_{max},\cdots,j_{max}\}}C_j (\sigma_{r}, x_{\pi(r,1)}, x_{\pi(r,2)})    \bigg]. 
\ea
\end{definition}

\begin{lemma}\label{lem6}
Fix $\rho$ as in Theorem \ref{thetheorem}.
In the cardioid defined by Figure \ref{cardio}, the series \eqref{eq62} is absolutely convergent. It is moreover uniformly convergent in $j_{\max}$. Its ultraviolet limit $\mathfrak{L}_{j_{\max}}(\lambda,\sigma)$ is therefore well-defined and
analytic in that cardioid domain.
\end{lemma}
\prf
For the proof, see \cite{RiZh2}. It is simply a corollary of the main result of \cite{RiZh2}, which is Theorem \ref {thetheorem} in this article.
\qed

\begin{lemma}\label{lem7}
\bea\label{E64}
\Big|  \sum_{\bar \cJ_\fm}  \int dw_{\bar \cJ_\fm}  \int d\nu_{\bar \cJ_\fm} \partial_{\bar \cJ_\fm}\Big| 
\times \Big| \mathfrak{M}^{2\fm}_{j_{\max}}(x_1, ..., x_{2\fm})\Big|  &\le& 8^{\fm}.
\ea
\end{lemma}
\prf
We validate this lemma by induction. 
\begin{itemize}
\item{$ {\mathcal K}=1$, i.e. for a cumulant of order 2}
\bea\label{EQ45}
\Big|\sum_{\bar \cJ_1}  \int dw_{\bar \cJ_1}  \int d\nu_{\bar \cJ_1}\Big|
\times \Big| \mathfrak{M}^{2}_{j_{\max}}(x_1, x_{2})\Big|
\le \Big|\sum_{j=0}^{j_{max}} C_j (\sigma_{1}, x_{1}, x_{2})  \Big|=1.\label{EQ46}
\ea
For the proof, this means that
$d\nu_{\bar \cJ_1} (\{\sigma_v\}, \sigma_1, \{ w \}) $  has covariance  
\bea\label{EQ53}
<\sigma_\alpha,\sigma_{\alpha'}>=
\delta (x-y) w^{\bar \cJ_1}  (\alpha, \alpha', \{ w\})
\ea
 where $w^{\bar \cJ_1}  (\alpha, \alpha', \{ w\})$ is 1 if $\alpha=\alpha'$ (where $\alpha, \alpha'\in \{v\}, 1$),
and the infimum of the $w_\ell$ for $\ell$ running over the unique path from $\alpha$ to $\alpha'$ in ${\bar \cJ_1} $
if $\alpha\ne \alpha'$. Now  
it suffices to prove  that $x_1,x_2$ does not depend on $j_{max}$. When $\alpha, \alpha'\in \{v\}$, it results from  \cite{RiZh2} and the notations of Theorem \ref{thetheorem}. When $\alpha=\alpha' =1$, combine the fact that the external face is renormalized, i.e., $w^{\bar \cJ_1}$ is $1$. When $\alpha\neq \alpha' $ and $\alpha=1$ or $\alpha'=1$, it result from 
$w^{\bar \cJ_1}  (\alpha, \alpha', \{ w\})\le 1$, which is proved in \cite{{RiZh2}}.

\item{$ [{\mathcal K} \le s] \Rightarrow   [{\mathcal K}=s+1]$}
\end{itemize}

- Let us first treat the case of one broken face with $s+ 1$ cilia and the rest of unbroken faces. 
The number of the left part and the right part of the cilia 
must be equal no matter where the gluing of the left part and the right part of the cilia is, since the total number of the cilia is trapped by the single broken face. Therefore, the number of left half-parts of the cilia must be $s+1$,
and so must the right half-parts; the number of the right half-parts of the cilia must also be $s+1$. 
Simply combine the combinatorics of \eqref{EQ55}-\eqref{EQ56}
with the fact that there is the left part and the right part of the cilia in this model, see \eqref{E39}.
We arrive at 
\bea\label{E64.0}
\Big|  \sum_{\bar \cJ_\fm}  \int dw_{\bar \cJ_\fm}  \int d\nu_{\bar \cJ_\fm} \partial_{\bar \cJ_\fm}\Big| 
\times \Big| \mathfrak{M}^{2\fm}_{j_{\max}}(x_1, ..., x_{2\fm})\Big|  &\le& 8^{\fm}. 
\ea 
Then Lemma \ref{lem7} holds for this case, i.e. ${\mathcal K}=s+1$.

 - Finally, let us treat the more difficult case corresponding to one broken face with $1\le {s}_{1}\le  s$ cilia, the total number of the other cilia being $s+1-{s}_{1}$.  Then the cilia factorize between ${s}_{1}$ and the rest $s+1-{s}_{1}$. 
 We are taking into account the fact that the Lemma \ref{lem7} is already proven for $ {\mathcal K} \le s_1$ and $ {\mathcal K} \le {s+1-s_1}$, therefore:
\bea\label{E68}
&&\Big|  \sum_{\bar \cJ_{s_1}}  \int dw_{\bar \cJ_{s_1} } \int d\nu_{\bar \cJ_{s_1}}
\partial_{\bar \cJ_{s_1}}\Big| \times \Big| \mathfrak{M}^{2s_{j_{\max}}}(x_1, ..., x_{2s_1}) \Big| \le 8^{s_1} ,
\\&&\nonumber
\Big|  \sum_{\bar \cJ_{s+1-s_1}}  \int dw_{\bar \cJ_{s+1-s_1}}  \int d\nu_{\bar \cJ_{s+1-s_1}} \partial_{\bar \cJ_{s+1-s_1}}\Big| 
\\&&
\label{E69}
\hskip3cm \times \Big| \mathfrak{M}^{2{s+1-s_1}}_{j_{\max}}(x_1, ..., x_{2{s+1-s_1}})\Big| \le 8^{s+1-s_1} .
\ea
Therefore Lemma \ref{lem7} holds for this case, i.e. ${\mathcal K}=s+1$.
\qed

\begin{lemma}\label{lem9}
Fix $\rho$ as in Theorem \ref{thetheorem}. In the cardioid defined by Figure \ref{cardio1} 
the series  $\mathfrak{K}^{2\fm}_{j_{\max}}(\lambda,x_1, ..., x_{2\fm})$
 is absolutely convergent. It is moreover uniformly convergent in $j_{\max}$. Its ultraviolet limit 
 $\mathfrak{K}^{2\fm}(\lambda,x_1, ..., x_{2\fm})$ is therefore well-defined and
analytic in that cardioid domain.
\end{lemma}
\prf
Simply combine Lemma \ref{lem6} and Lemma \ref{lem7}.
\qed

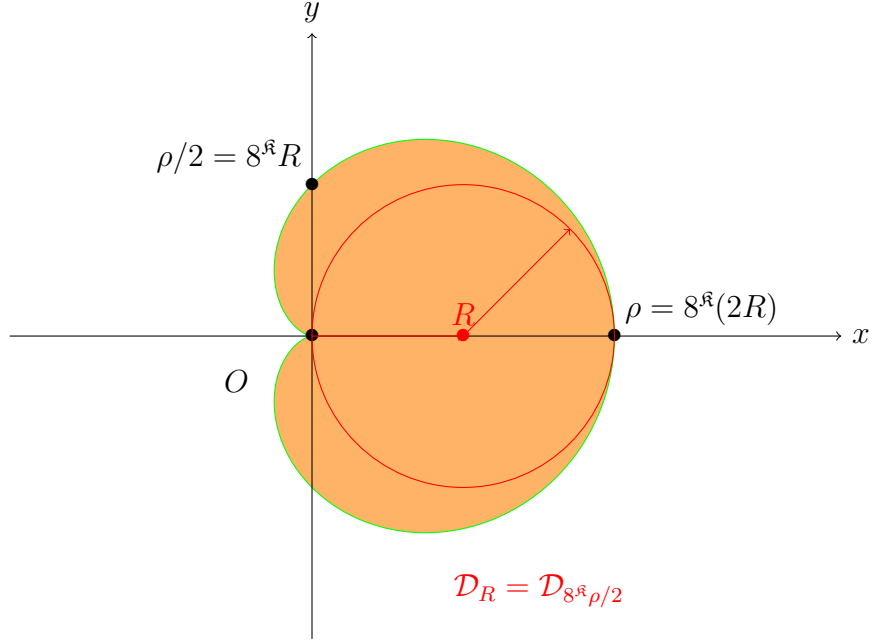
\begin{figure}[!ht]
\begin{center}
\begin{tikzpicture}[scale=2]
\tikzstyle{red2}=[thick,  fill=red!40]
\draw (-0.5,-0.3) node {$O$} ;
\draw (2,0) node[above right]{$\rho=8^\mathfrak{K}(2R)$} ;
\draw (0,1) node[above left]{$\rho/2=8^\mathfrak{K} R$} ;
\filldraw [draw=green,fill=orange!60,domain=-180:180,samples=400,scale=.5,variable=\t] plot (\t:  2+ 2*cos \t);
\draw (-0,0) node {$\bullet$} ;
\draw (2,0) node {$\bullet$} ;
\draw (0,1) node {$\bullet$} ;
\draw[->] (-2,0) -- (3.5,0);
\draw (3.5,0) node[right] {$x$};
\draw [->] (0,-2) -- (0,2);
\draw (0,2) node[above] {$y$}; 
\draw[red] (1,0) circle (1) ;
\draw[red] (1,0) node {$\bullet$} ;
\draw[red] (1,0) node[above]{$R$} ;
\draw[red,->] (0,0) -- (1,0);
\draw[red,->] (1,0) -- (1.7071,0.7071);
\draw[red] (1.5,-1.5) node[below]{$\cD_R=\cD_{8^\mathfrak{K}\rho/2}$} ;
\end{tikzpicture}
\end{center}
\caption{Open cardioid domain pictured in orange, frontier pictured in green. In red, $D_R$ for $R=8^\mathfrak{K}\rho/2$, see Appendix 
B.}
\label{cardio1}
\end{figure}

\medskip
Finally, for the proof of Theorem \ref{THM1},  all that remains is the proof of Borel summability.
For that, let us take $f(\lambda, {\cal N})$ as $\mathfrak{K}^{2\fm}(\lambda,x_1, ..., x_{2\fm})$ and 
apply Appendix B for  ${\cal N} \in
[0,1]^{2\fm}$ and $ \mathfrak{K}^{2\fm} (\lambda,x_1, ..., x_{2\fm})= f(\lambda,{\cal N} )$.

\qed

 \section{Outlook}

Where renormalization is present, we are interested
in a constructive Borel summability statement about the sources $J$ in \eqref{eqp1}.
It would be desirable to extend this article to include topological development and Weingarten calculus  \cite{Col}.

\medskip
In the more distant future, we are keen to explore the potential of providing additional models for cumulants, including the $\phi^4_3$ model \cite{GJ,MS},
$T^4_3$ model \cite{DR} or $T^4_4$ \cite{RV}  model, which are all super-renormalisable.

\medskip
Recently, the subject of cumulants in probability from a combinatorial point of view 
has attracted increasing interest; see \cite{EFKP,celestino2022cumulant}.
We think that this MLVE with cumulants (MLVEC) is interesting from that point of view.
We also think it can be extended, for example, to various groups such as the $O(N)$ and $Sp(N)$ groups.
We hope in addition that it can be used in renormalisable field models, for example  in $\phi^4_4$ and $T^4_5$ \cite{BenGeloun:2011rc,KopWan1,KopWan2,RV1}
and in the formalism of Ecalle's transseries \cite{C,BGKL}. 
For an recent introduction to the constructive theory, renormalisation and
main themes of this paper, 
we recommend a new article \cite{BeChGu}.

\section{Appendix A}\label{sec60}

No paper about the LVE should forget an appendix about the BKAR formula
\cite{BK,AR1} since this formula is crucial to the LVE.
Any quantity $F$ in quantum field theory which is an integral over a Gaussian complex measure 
can be combinatorially represented as a sum over the set $\cal F$.

\begin{lemma}[BKAR formula] \label{thm:BKAR}

Let $f$ be a smooth function of $n(n-1)/2$ line variables $x_\ell \in [0,1]$, $\ell = (i,j)$, $1 \le i< j\le n$. The forest formula states
\bee f(1,\cdots,1)= \sum_{\cal F}\big\{\prod_{\ell \in {\cal F}} [\int_0^1 dw_{\ell}] \big\}\big\{\prod_{\ell \in {\cal F}}\frac{\partial}{\partial x_{\ell}} f \big\}[X^{{\cal F}}(w_{{\cal F}})], \quad \text{where} \nonumber
\ee
\begin{itemize}
\item the sum over ${\cal F}$ is over all forests over $n$ vertices,

\item the ``weakening parameter" $X_{ij}^{{\cal F}}(w_{{\cal F}})]$ is $0$ if $i$ and $j$ dont belong to
 the same connected component of $\cal F$; otherwise it is the
 minimum of the $w_{\ell '}$ 
for $\ell'$ running over the unique path from $i$ to $j$ in $\cal F$.

\item Furthermore the real symmetric matrix $X_{ij}^{{\cal F}}(w_{{\cal F}})$ (completed by $1$ on the diagonal $i=j$)
is positive.
\end{itemize}
\end{lemma}
\proof See \cite{BK,AR1}.
\qed

\medskip
For readers who want to look further into the BKAR formula and \emph{oriented forests, ordered or not}, see \cite{RiZh1,RiTa}.
\section{Appendix B}\label{sec6}

We recall the Nevanlinna-Sokal theorem \cite{Nev,Sokal} \footnote{
For Borel-LeRoy modifications, see \cite{CaGrMa}.}.
Here we follow the notations of \cite{Gur}, Appendix C:

\begin{figure}[!ht]
\begin{center}\begin{tikzpicture}[scale=1.2]
\fill[color=red!40] (1,0) circle (1) ;
\draw[->] (-1,0) -- (2.5,0);
\draw [->] (0,-1.5) -- (0,1.5);
\draw[red] (1,0) circle (1) ;
\draw[red] (1.1,.7) node[below]{$^{R}$} ;
\draw[red] (1.5,-1.5) node[below]{$\cD_R$} ;
\draw[red,<->] (1,0) -- (1.7071,0.7071);
\draw[red] (1,0) circle (1) ;
\fill[color=red!40] (5,1) arc (90:270:1) -- (9,-1) -- (9,1) -- cycle ;
\draw[->] (3.5,0) -- (9.5,0);
\draw [->] (5,-1.5) -- (5,1.5);
\draw[red] (5,1) arc (90:270:1) -- (9,-1) -- (9,1) -- cycle ;
\draw[red,<->] (5,0) -- (4.2929,0.7071);
\draw[red] (4.4,0.5) node[below]{$^{\sigma^{-1}}$} ;
\draw[red] (6,-1.5) node[below]{$\Sigma_\sigma$} ;
\end{tikzpicture}
\end{center}
\caption{Domain of analyticity of $F$ and of its Borel transform.}
\label{CuX}
\end{figure}
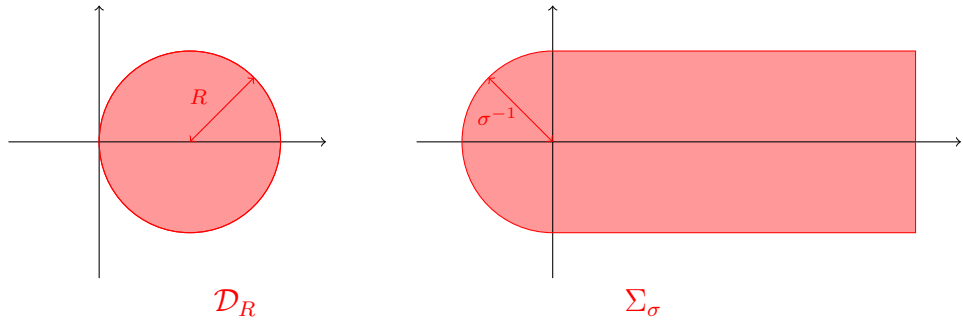

\begin{theorem}
\label{BorLeRSum}
A function $f(\lambda, {\cal N})$ with $\lambda \in \mathbb{C}$ and ${\cal N} \in [0,1]^{2\fm}$
is said to be Borel summable in $\lambda$ uniformly in ${\cal N}$ if:

\begin{itemize}
\item $f(\lambda, {\cal N})$ is analytic in a disk \,
$\Re (\lambda^{-1}) > (2R)^{-1}$ with $R\in \mathbb{R}_+$ independent of ${\cal N}$.
\item $f(\lambda, {\cal N})$ admits a Taylor expansion at the origin with uniform bound on the Taylor remainder:
\bea
f(\lambda, {\cal N})= \sum_{k=0}^{r-1} f_{{\cal N},k} \lambda^k + R_{{\cal N},r}(\lambda),\quad\vert R_{{\cal N},r}(\lambda)\vert \le K \sigma^r r! \vert
\lambda \vert^r ,
\ea
for some constants $K$ and $\sigma$ independent of ${\cal N}$.
\end{itemize}
If $f(\lambda, {\cal N})$ is Borel summable in $\lambda$ uniformly in ${\cal N}$ then:
\bea B(t,{\cal N}) &=&  \sum_{k=0}^{\infty} \frac{1}{k!}\:  f_{{\cal N},k}\, t^k ,
\ea
is an analytic function for $\vert t \vert < \sigma^{-1}$ that admits an analytic continuation in the strip
$\{z |\,  \vert  \Im z \vert < \sigma^{-1}\}$such that $\vert B(t,{\cal N}) \vert \le B\,e^{t/R}$ for some constant $B$
independent of ${\cal N}$ and $f(\lambda, {\cal N})$ is given by the absolutely convergent integral:
\bea
f(\lambda, {\cal N}) = \frac{1}{\lambda} \int_0^\infty dt B(t,{\cal N}) e^{- \frac{t}{\lambda}}\, .
\ea
\end{theorem}
In other words, the Taylor expansion of $f(\lambda, {\cal N})$ at the origin is Borel summable, and $f(\lambda, {\cal N})$ is its Borel sum.

\end{document}